%% file: index.tex
\documentclass[submission,copyright,creativecommons]{eptcs}
\providecommand{\event}{EXPRESS/SOS 2012} 

\usepackage{amsmath,amssymb,amsfonts,amsthm}
\usepackage{xspace}
\usepackage{xcolor}

\usepackage{tikz}
\usetikzlibrary{arrows,decorations.pathmorphing,positioning,petri}
\tikzset{
  node distance=40pt,
  every place/.style=
    {
      circle,
      draw,
      thick,
      inner sep=3pt,
      minimum size=6mm
    },
  every transition/.style=
    {
      rectangle,
      draw,
      thick,
      inner sep=3pt,
      minimum size=6mm
    },
  edge/.style=
    {
      ->,
      shorten <=1pt,
      >=stealth',
      semithick
    },
  inhibitor/.style=
    {
      -o,
      shorten <=1pt,
      >=stealth',
      semithick
    },
  dedge/.style=
    {
      <->,
      shorten <=1pt,
      >=stealth',
      semithick
    },
  state/.style=
    {
      circle,
      fill=black,
      inner sep=.1pt,
      minimum size=4pt
    }
}

\newcommand\eg{e.\,g.,\xspace}
\newcommand\ie{i.\,e.,\xspace}
\newcommand\cf{cf.\xspace}
\newcommand\etal{et al.\xspace}

\newcommand\PN{Petri net\xspace}
\newcommand\PNs{Petri nets\xspace}
\newcommand\LTS{\ensuremath{\mathsf{LTS}}\xspace}

\newcommand\Jcore{\ensuremath{\mathcal{J}_{core}}\xspace}

\newcommand\sect[1]{Sect.~\ref{#1}\xspace}

\newcommand\fig[1]{Fig.~\ref{#1}\xspace}
\newcommand\ex[1]{Example \ref{#1}\xspace}

\newcommand\defn[1]{Definition~\ref{#1}\xspace}

\newcommand\lemm[1]{Lemma~\ref{#1}\xspace}
\newcommand\prop[1]{Proposition~\ref{#1}\xspace}

\newcommand\pre[1]{\ensuremath{{}^\bullet #1}\xspace}
\newcommand\post[1]{\ensuremath{#1^\bullet}\xspace}

\newcommand\defset{\ensuremath{\mathcal{D}}\xspace}

\newcommand\fv{\ensuremath{\mathsf{fv}}\xspace}
\newcommand\dv{\ensuremath{\mathsf{dv}}\xspace}
\newcommand\rv{\ensuremath{\mathsf{rv}}\xspace}

\newcommand\Mstruct{\ensuremath{\mathsf{M}}\xspace}

\newcommand\push{\ensuremath{\downarrow}\xspace}
\newcommand\pop{\ensuremath{\uparrow}\xspace}

\newcommand\id{\ensuremath{\mathsf{id}}\xspace}

\newcommand\xlongmapsto[1]{\ensuremath{\overset{#1}{\longmapsto}}\xspace}
\newcommand\xlongrightarrow[1]{\ensuremath{\overset{#1}{\longrightarrow}}\xspace}

\newcommand\nil{\ensuremath{\mathsf{0}}\xspace}
\newcommand\jdef{\,\ensuremath{\mathsf{def}}\,}
\newcommand\jin{\,\ensuremath{\mathsf{in}}\,}
\newcommand\join{\,\ensuremath{\triangleright}\,}
\newcommand\angled[1]{\ensuremath{\langle #1 \rangle}\xspace}

\newcommand\parl{\ensuremath{\,|\,}}
\newcommand\parll{\ensuremath{\,|\,}}

\theoremstyle{plain}
\newtheorem{theorem}{Theorem}
\newtheorem{lemma}{Lemma}
\newtheorem{proposition}{Proposition}
\newtheorem{corollary}{Corollary}
\theoremstyle{definition}
\newtheorem{definition}{Definition}
\newtheorem{example}{Example}

\title{An Operational Petri Net Semantics for the Join-Calculus}
\author{Stephan Mennicke
\institute{Institute for Programming and Reactive Systems\\
TU Braunschweig, Germany}
\email{mennicke@ips.cs.tu-bs.de}
}
\def\titlerunning{An Operational Petri Net Semantics for the Join-Calculus}
\def\authorrunning{S.~Mennicke}
\begin{document}
\maketitle

\begin{abstract}
We present a concurrent operational Petri net semantics for the join-calculus, a
process calculus for specifying concurrent and distributed systems.
There often is a gap between system specifications and the actual 
implementations caused by synchrony assumptions on the specification side and
asynchronously interacting components in implementations. The join-calculus
is promising to reduce this gap by providing an abstract specification
language which is asynchronously distributable.
Classical process semantics establish an implicit order of actually 
independent actions, by means of an interleaving. So does the semantics of the join-calculus. 
To capture such independent actions, step-based semantics, \eg
as defined on Petri nets, are employed. Our Petri net semantics for the
join-calculus induces step-behavior in a natural way.
We prove our semantics behaviorally equivalent to the original 
join-calculus semantics by means of a bisimulation. We 
discuss how join specific assumptions influence an existing notion of 
distributability based on Petri nets.
\end{abstract}

\input{introduction}
\input{preliminaries}
\input{semantics}
\input{discussion}
\input{conclusion}

\bibliographystyle{eptcs}
\bibliography{references}


\end{document}

%% file: introduction.tex
\section{Introduction}\label{sect:introduction}
Specifications for distributed systems
usually employ synchrony assumptions to keep the modeling as simple as
possible. Properties of specifications cannot be reused for real 
implementations, because components in a distributed system run concurrently and
communicate in an asynchronous fashion. This leaves a gap between
specifications and implementations.

Process
calculi, \eg the $\pi$-calculus, concentrate
on the essential parts in system specifications, keeping in mind that
they represent actual systems. Therefore, they come with a syntax and
a semantics to describe the behavior of a system as precise as possible.
The asynchronous $\pi$-calculus, a restricted $\pi$-calculus, tries
to reduce the gap between system specifications and implementations.
By the asynchronous $\pi$-calculus, we are able describe asynchronously communicating systems, but implementations
still rely on hard to implement constructs, such as {\em rendezvous}
or {\em leader election} \cite{Milner1992}.

The join-calculus by Fournet and Gonthier \cite{Fournet1996}
is a process calculus equipped with a basic language and an abstract
notion of computation, the {\em reflexive chemical abstract machine}.
Fournet and Gonthier extend Berry and Boudol's {\em chemical
abstract machine} \cite{Berry1990} by explicit reaction sites -- similar to {\em locations} in
distributed systems -- and combine the concepts of restriction, reception
and recursion in one construct called a {\em join definition}. By join
definitions, they force receptors, \ie names which are used to receive
messages, to reside on one location. In contrast, $\pi$-calculus allows
the use of sent names as receptors (\cf {\em scope extrusion}) which
enables the calculus to describe the concept of {\em mobility}, but
makes distributed implementations of the calculus difficult.

Still, as many other process calculi, the join-calculus only comes with
an interleaving semantics which makes it hard to reason about
the distributed behavior of processes. Although the join-calculus is equipped
with a parallel composition operator, it is rather difficult to describe
independence of actions, whereas other models, such as Petri nets \cite{Petri1962}, describe
independence explicitly. Therefore, we present an operational Petri net
semantics for the join-calculus taking advantage of the parallel structure to
obtain a large degree of independence, \ie concurrency.

The general idea
of our Petri net semantics is inspired by the work of
Busi and Gorrieri \cite{Busi2009}, where they propose a Petri net
semantics with inhibitor arcs for the $\pi$-calculus. They decompose
a $\pi$ term into places and construct the nets by transition rules working
on decompositions. They solve scoping issues in the $\pi$-calculus by a
global renaming. Our semantics does not rely on such a renaming as we store
the message scopes in places. As in Busi and Gorrieri's semantics,
all necessary information is encoded in the initial decomposition
corresponding to initially marked places. We concentrate on the
core join-calculus which is not equipped with an explicit choice. Therefore, we
can also abandon inhibitor arcs from our semantics. In general,
our semantics yields infinite but 1-safe Petri nets. It also comes with a bisimulation result
to the original join-calculus semantics ensuring the correctness of our approach.

\PNs and \PN related formalisms have already
been used to describe the semantics of the join-calculus. Buscemi and
Sassone propose a type-theoretic approach by
suggesting a hierarchy on the syntax of the join-calculus \cite{Buscemi2001}. For each level,
they prove that, if a join-calculus term is typable, \ie is satisfying
a restriction on the syntax, then the Petri net of the join term they construct is
bisimilar to the original join-calculus semantics. They get place/transition
nets by restricting processes to top-level join definitions. To handle more
expressive join terms, they use colored, reconfigurable and dynamic Petri nets.
In our work, we cover full expressiveness of the join-calculus by an infinite construction.
Bruni \etal propose an event structure semantics for
the join-calculus \cite{Bruni2006}. Their main goal is to establish so called
{\em persistent graph grammars} as a tool to describe
name passing process calculi. They focus on an encoding from the
asynchronous $\pi$-calculus into persistent graph grammars. The unfolding
of the grammars yields event structures. For the join-calculus 
they yield event structures with empty concurrency relations. 
The semantics we propose includes concurrency by exploiting the parallel structure of a join term.
There are also more general approaches which do not give a semantics
for the join-calculus, but use the same ideas to obtain new Petri net classes.
Prominent examples are {\em mobile and dynamic Petri nets} by Asperti and
Busi \cite{Asperti2009} and {\em functional nets} by Odersky \cite{Odersky2000}.
Our approach does not aim at extending Petri nets or introducing new
extensions to \PN theory.

Unfortunately, our net semantics yields infinite nets which seems to make
it impossible to be useful for any real-world applications. Due to
nice structural properties of the nets, the semantics could be directly used
for any {\em unfolding based} techniques on Petri nets. One of such applications
is model-checking. In Petri net unfoldings \cite{Esparza2001}, it is not necessary to compute the potentially
infinite structure of the net, but make use of a finite representation called
{\em prefix}. In this paper, we want to investigate the join-calculus
in terms of distributability. Recent research \cite{Glabbeek2008,Schicke2011} suggest a
notion of distributed systems in terms of Petri nets and proved a
Petri net structure, which refers to symmetric confusion, to
be impossible to distribute. If our proposed semantics is reasonable
and correct, we may argue on the distributability of the calculus itself.

The rest of the paper is structured as follows.
\sect{sect:preliminaries} introduces the 
necessary notions for this paper including \PNs
(\sect{sect:petrinets}) and an overview of the join-calculus
(\sect{sect:join}). The following section
is concerned with the definition of our Petri net semantics for the
join-calculus and its correctness results. In \sect{sect:distributability},
we discuss a notion of distributability and how the join-calculus
influences it. In \sect{sect:conclusion}, we conclude our work and give 
some further research directions.

%% file: preliminaries.tex
\section{Preliminaries}\label{sect:preliminaries}
In this section, we introduce the basic notions and concepts used in
our net semantics. First, we need the notion of multisets.

\begin{definition}[Multisets]
Let $A$ be a set. A {\em multiset $M$ over $A$} is a mapping from $A$ to $\mathbb N$.
For $a \in A$, $M(a) = 0$ iff $a \not\in M$. Otherwise $a \in M$. Two multisets
$M_1, M_2$ over $A$ can be unified by $\uplus$. $M_1 \uplus M_2$ is a
multiset where for each $a \in A$, $(M_1 \uplus M_2)(a) = M_1(a) + M_2(a)$.
\end{definition}


Whenever $f$ is a function from a set $A$ to a cartesian product $\prod_{i = 0}^n A_i$,
then we define the projections on the result of $f$ by $f^i := \pi_i \circ f$,
where $\pi_i$ is the projection function on the $i$th component of the
product.
\id denotes the {\em identity function} defined on any set.

In our semantics we need to store scopes for objects.
These scopes may be nested. To handle this nesting of scope we introduce
the notion of {\em stacks} -- a common data structure also used in compilers.
A stack may be empty ($\bot$) or filled with elements of an alphabet.
It is equipped with three operations. First, the {\em push} operation
adds an element on top of a stack. Second, the {\em top} operation
returns the top element of a stack. Last, the {\em pop} operation
removes the top element of a stack.

\begin{definition}[Stack]
Let $\Sigma$ be an alphabet. A {\em stack $s$ over $\Sigma$} is either 
$\bot$ or $s$ contains at least one element $e \in \Sigma$, 
\ie, $s = [ e, s' ]$, where $s'$ is a stack over $\Sigma$. The
set of all stacks over $\Sigma$ is denoted by $\mathcal S_\Sigma$.
The following operations are defined on $\mathcal S_\Sigma$.
\begin{itemize}
\item $\top : \mathcal S_\Sigma \to \Sigma$ denotes the top element of a stack $s$
with
$$s\top := \left\{\begin{array}{ccl}
\varepsilon && s = \bot \\
e && s = [ e, s' ]\text.
\end{array}\right.$$
\item $\push : \mathcal S_\Sigma \times \Sigma \to \mathcal S_\Sigma$
denotes the push operation. For a stack $s$ and a symbol $e$,
$s\push e := [ e, s ]$.
\item $\pop : \mathcal S_\Sigma \to \mathcal S_\Sigma$ denotes the
pop operation. For a stack $s$,
$$s\pop = \left\{\begin{array}{ccl}
\bot && s = \bot \\
s' && s = [ e, s' ]\text.
\end{array}\right.$$
\end{itemize}
\end{definition}

Instead of $[ e_1, [ e_2, [ \dots, [ e_n, \bot ] \dots ]]$ we write
$[e_1, e_2, \dots, e_n, \bot]$.

Labeled transition systems serve as the common semantic model of both formalisms, Petri nets
and the join-calculus. It consists of
three components, a set of states $Q$, a labeled relation between states $\to$
and a start state $q_0$. The labels for so called transitions are obtained
from some alphabet $\Sigma$.

\begin{definition}[\LTS]
A {\em labeled transition system (over $\Sigma$), \LTS}
is a triple, $( Q, \to, q_0 )$ where $Q$ is a set, $\to \subseteq Q \times \Sigma \times Q$,
and $q_0 \in Q$.
\end{definition}

In \cite{Glabbeek2001}, van Glabbeek gives a huge
collection of behavioral equivalences for \LTS{}s. {\em Bisimulation} is a very
strong equivalence taking the branching structure, \ie the structure
of decisions, of a system into account. As already mentioned, Petri nets
as well as the join-calculus have an \LTS semantics. Therefore, we
introduce the notion of bisimulation. Later in \sect{sect:correctness}
we will prove our semantics introduced in \sect{sect:pnsemantics} to
be {\em bisimilar} to the original semantics of the join-calculus.

\begin{definition}[Bisimulation]
Let $A_1 = ( Q_1, \to_1, q_1 )$ and $A_2 = ( Q_2, \to_2, q_2 )$ be
labeled transition systems over some alphabet $\Sigma$. A relation $\mathcal R \subseteq Q_1 \times Q_2$
is called a {\em bisimulation between $A_1$ and $A_2$} iff
\begin{itemize}
\item $(q_1, q_2) \in \mathcal R$,
\item if $(p, q) \in \mathcal R$ and $p \xrightarrow{a}_1 p'$, then
there exists $q' \in Q_2$ such that $q \xrightarrow{a}_2 q'$ and
$(p', q') \in \mathcal R$, and
\item if $(p, q) \in \mathcal R$ and $q \xrightarrow{a}_2 q'$, then
there exists $p' \in Q_1$ such that $p \xrightarrow{a}_1 p'$ and
$(p', q') \in \mathcal R$.
\end{itemize}
If such a relation exists, then $A_1$ and $A_2$ are {\em bisimilar}.
\end{definition}

\subsection{Petri Nets}\label{sect:petrinets}
{\em Petri nets} were first introduced by Carl Adam Petri \cite{Petri1962}. Petri nets are directed bipartite
graphs with places drawn as circles and transitions drawn as boxes.
Places and transitions are the nodes of a net. Directed edges called arcs, either connect 
places with transitions or transitions with places.
An example is depicted in \fig{example:app_semantics}.
We assume a universe of places denoted by $\mathcal P$. 
We later specify $\mathcal P$ to meet the purposes of our semantics. 
The set of net places is a subset of $\mathcal P$.
As in labeled transition systems we have a fixed alphabet
$\Sigma$ for transition labels representing the actions of a system.
In contrast to classical net definitions, we directly encode the set of arcs
into transitions.
\begin{definition}[Net]
The tuple $N = ( P, T )$ is called a {\em labeled net over $\Sigma$}
iff
\begin{itemize}
\item $P \subseteq \mathcal P$ is a set and
\item $T \subseteq 2^P \times \Sigma \times 2^P$.
\end{itemize}
\end{definition}
The label $\pi_2(t)$ of a transition $t$ is also referred to as $l(t)$.
Here, $l$ is implicitly given and not a part of the net definition.
The preset of a transition $t$ is denoted by $\pre t := \pi_1(t)$,
the postset of $t$ is denoted by $\post t := \pi_3(t)$. Pre- and postsets
of places are defined by $\pre p := \{ t \in T \,|\, p \in \post t \}$
and $\post p := \{ t \in T \,|\, p \in \pre t \}$. The arc relation
is obtained by $F = \{ (p, t) \in P \times T \,|\, p \in \pre t \}
\cup \{ (t, p) \in T \times P \,|\, p \in \post t \}$.

A net is called {\em finite} iff $(P \cup T)$ is finite. Otherwise, the net
is called {\em infinite}.


The potential state of nets is described by {\em markings}, which are 
multisets over the set of places. Tokens, drawn as black dots 
(\cf \fig{example:app_semantics}), represent the number of places in a 
marking. These states may change by {\em firing} transitions.
Transitions are {\em enabled} iff there is at least one token on
any {\em input place} $p \in \pre t$. An enabled transition may fire, 
which means that it {\em consumes} one token from each input place and 
{\em produces} one token on any output place $p \in \post t$. This 
procedure is formally defined by the {\em firing rule}.
\begin{definition}[Enabledness, Firing rule]
Let $N = ( P, T )$ be a net and let $m : P \to \mathbb N$ be a marking of $N$. A transition
$t \in T$ is {\em enabled under $m$}, written $m[t\rangle$, iff $m(p) > 0$ for all $p \in \pre t$.
An enabled transition $t \in T$ may {\em fire}. The {\em successor marking of $m$} by
firing $t$ is $m'$, written $m[t\rangle m'$, with
$$m'(p) = \left\{\begin{array}{ccccl}
m(p) & + & 1 && \text{if } p \in ( \post t \setminus \pre t ) \\
m(p) & - & 1 && \text{if } p \in ( \pre t \setminus \post t ) \\
m(p) &&&& \text{else.}
\end{array}\right.$$
\end{definition}
{\em Petri nets} are nets with an initial marking $m_0$ corresponding to
the start state of a net.
\begin{definition}[Petri net]
The triple $N = ( P, T, m_0 )$ is called a {\em Petri net} iff
$( P, T )$ is a net and $m_0$ is a marking of $( P, T )$.
\end{definition}

A marking $m$ is {\em reachable} in a net $N = ( P, T, m_0 )$ iff 
there exists a sequence of transitions $t_1, \dots, t_n$ ($t_i \in T$) such that 
$m_0 [ t_1 \rangle \dots [t_n \rangle m$.
The set of all reachable markings of $N$ is denoted by $Reach(N)$. By relating
reachable markings we derive an \LTS from a Petri net.
\begin{definition}{}
Let $N = ( P, T, m_0 )$ be a Petri net labeled over $\Sigma$.
The \LTS of $N$ is defined by $\LTS(N) := ( Reach(N), \xlongrightarrow{}, m_0 )$ where
\begin{itemize}
\item $\xlongrightarrow{} \subseteq Reach(N) \times \Sigma \times Reach(N)$ and
\item $( m, a, m' ) \in \xlongrightarrow{}$ iff there exists $t \in T$ with 
$m [ t \rangle m'$ and $l(t) = a$.
\end{itemize}
\end{definition}
Instead of $(m, a, m') \in \xlongrightarrow{}$ we often use the abbreviation
$m \xlongrightarrow a m'$.

\subsection{Join-Calculus}\label{sect:join}
The join-calculus \cite{Fournet1996} is a process algebra describing
the model of the {\em reflexive chemical abstract machine} based on Berry's and
Boudol's chemical abstract machine \cite{Berry1990}.
One of the reasons for the development of the join-calculus was the
difficulty to actually implement distributed CCS or distributed
$\pi$-calculus.
In comparison to the $\pi$-calculus by Milner \cite{Milner1992}, the
join-calculus combines restriction, recursion and reception in one 
construct called {\em join definition}, forcing receptors to reside
on one {\em location}. Hence, it is not possible to {\em extrude} a name and
use the same name for reception. Fournet and Gonthier \cite{Fournet1996} identified
a strict subset of the join-calculus which is
proven to be as expressive as the full calculus. This subset is called
{\em core join-calculus}. This section and our Petri net semantics is 
based on the core calculus.

For further notions, we assume an infinite set of names $\mathcal{N}$.
The syntax of the core join-calculus is defined in \fig{join:syntax}.
\nil stands for the {\em null process}, a process with no behavior. 
$x\angled{v}$ represents {\em output messages}.
As in the $\pi$-calculus, $x$ stands for the channel name and $v$ is a value
passed through $x$. The {\em parallel composition} of two processes $P$ 
and $Q$ is denoted by $P \parll Q$, where $P$ and $Q$ work independently. The last
syntactic element is the {\em definition}, $\jdef x\angled{u} \parll y\angled{v} \join Q \jin P$.
Definitions combine restriction, reception of names and recursion in one construct.
$x\angled{u} \parll y\angled{v}$ is called the {\em join-pattern}. 
$x\angled{u} \parll y\angled{v} \join Q$ is called the {\em join definition}
or, together with a process $P$, the {\em enclosing definition} of $P$. We denote the
set of all join definitions by $\defset$. $P$ is the {\em enclosed process}. 
The set of all core join terms is denoted by \Jcore.

\begin{figure}
\begin{center}
\begin{minipage}{13cm}
$$\begin{array}{rcc|c|c|c}
P & ::= & ~ \nil ~ & ~ x\angled{v} ~ & ~ P \parll P ~ & ~ \jdef x\angled{u} \parll y\angled{v} \join P \jin P
\end{array}$$
\end{minipage}
\end{center}
\caption{Syntax of \Jcore}\label{join:syntax}
\end{figure}


Variables in a join-term are partitioned into three categories which
are not necessarily disjoint. The {\em free variables} ($\fv$) are those
being visible to the environment. {\em Defined variables} ($\dv$) are
variables bound to a join definition, \ie those channels that are processed
by a definition. {\em Received variables} ($\rv$) are only locally bound
to new processes resulting from the application of join definitions.
These three sets are defined in \fig{join:vars} (\cf \cite{Fournet1998}).

\begin{figure}[p]
\begin{center}
\begin{minipage}{13cm}
$$\begin{array}{rcl}
\fv[ x\angled{v_1, \dots, v_n} ] & \overset{\mathsf{Def}}= & \{ x, v_1, \dots, v_n \} \\
\fv[ \jdef D \jin P ] & \overset{\mathsf{Def}}= & ( \fv[P] \cup \fv[D] ) \setminus \dv[D] \\
\fv[ P \parll P' ] & \overset{\mathsf{Def}}= & \fv[P] \cup \fv[P'] \\
\fv[ 0 ] & \overset{\mathsf{Def}}= & \emptyset \\\\
\fv[ {J} \join {P} ] & \overset{\mathsf{Def}}= & \dv[J] \cup ( \fv[P] \setminus \rv[J] ) \\
\dv[ {J} \join {P} ] & \overset{\mathsf{Def}}= & \dv[J]
\end{array}$$
$$\begin{array}{rclcrcl}
\dv[ x\angled{y_1, \dots, y_n} ] & \overset{\mathsf{Def}}= & \{ x \} && \dv[ J \parl J' ] & \overset{\mathsf{Def}}= & \dv[J] \uplus \dv[J'] \\
\rv[ x\angled{y_1, \dots, y_n} ] & \overset{\mathsf{Def}}= & \{ y_1, \dots, y_n \} && \rv[ J \parl J' ] & \overset{\mathsf{Def}}= & \rv[J] \uplus \rv[J']
\end{array}$$
\end{minipage}
\end{center}
\caption{Free, defined and received variables of \Jcore terms}\label{join:vars}
\end{figure}

We use $\sigma_\fv, \sigma_\dv, \sigma_\rv$ to denote a renaming on the set
of free, defined and received variables.

The core join-calculus has its roots in an abstract machine called the
{\em reflexive chemical abstract machine}. Instead of specifying a set
of reduction rules, the chemical abstract machine first defines a structural
congruence and, on top of that, there is only one reduction rule. In
process calculi this method is adopted to reduce the number of rules for a
structural operational semantics significantly. As we want to use the
structural operational semantics to define labeled transition systems
of the core join-calculus, we first need the structural congruence
of core join terms. The congruence defined in \fig{join:struct} is 
reduced to the core join-calculus (\cf \cite{Fournet1998}).

\begin{figure}
\begin{center}
\begin{minipage}{13cm}
$$\begin{array}{rclll}
P \parl 0 & \equiv & P \\
P \parl Q & \equiv & Q \parl P \\
( P \parl Q ) \parl R & \equiv & P \parl ( Q \parl R ) \\
P \parl \jdef D \jin Q & \equiv & \jdef D \jin {P \parl Q} && \text{if } \fv(P) \cap \dv(D) = \emptyset \\
\jdef {D} \jin {\jdef {D'} \jin {P}} & \equiv & \jdef {D'} \jin {\jdef {D} \jin {P}} && \text{if } \fv(D) \cap \fv(D') = \emptyset\\\\
\jdef {D} \jin {P} & \equiv & \jdef {D\sigma_\dv} \jin {P\sigma_\dv} && \text{if } \sigma_\dv \text{ injective}\\
\jdef {D} \jin {P} & \equiv & \jdef D\sigma_\rv \jin P && \text{if } \sigma_\rv \text{ injective}\\
\end{array}$$
\end{minipage}
\end{center}
\caption{Structural congruence on \Jcore}\label{join:struct}
\end{figure}

\begin{figure}
\begin{center}
\begin{minipage}{13cm}
\begin{center}
\begin{tabular}{lrcl}
($\textsc{Join}$) & $x\angled{s} \parl y\angled{t}$ & $\xrightarrow{x\angled u \parl y\angled v \join R}$ & $R[s/u, t/v]$ \\\\
($\textsc{React}$) & \multicolumn{3}{c}{$\dfrac{P \xrightarrow D P'}{\jdef D \jin P \xlongmapsto{D} \jdef D \jin P'}$} \\\\
\multicolumn{2}{c}{($\textsc{Par1}$) $\dfrac{P \xrightarrow D P'}{P \parl Q \xrightarrow D P' \parl Q}$} & \multicolumn{2}{c}{($\textsc{Par2}$) $\dfrac{P \xlongmapsto{D} P'}{P \parl Q \xlongmapsto{D} P' \parl Q}$} \\\\
\multicolumn{4}{c}{($\textsc{Jump1}$) $\dfrac{P \xrightarrow D P', \dv(D) \cap \fv(D') = \emptyset}{\jdef D' \jin P \xrightarrow D \jdef D' \jin P'}$} \\\\
\multicolumn{4}{c}{($\textsc{Jump2}$) $\dfrac{P \xlongmapsto{D} P'}{\jdef D' \jin P \xlongmapsto{D} \jdef D' \jin P'}$} \\\\
\multicolumn{2}{c}{($\textsc{Struct1}$) $\dfrac{P \xrightarrow D P', P \equiv Q}{Q \xrightarrow D Q'}$} & \multicolumn{2}{c}{($\textsc{Struct2}$) $\dfrac{P \xlongmapsto{D} P', P \equiv Q}{Q \xlongmapsto{D} Q'}$}
\end{tabular}
\end{center}
\end{minipage}
\end{center}
\caption{Labeled transition semantics of \Jcore}\label{join:semantics}
\end{figure}

From the structural congruence we observe that it does not matter what the exact
defined variables are. In consequence, we may rename them. We
thereby need to make sure that all occurrences of defined variables in the
enclosed process are renamed as well. Later, our semantics will keep track
of definitions. To make sure that there are no name clashes, we introduce
a minimal notion of {\em normality} on which we rely. Our normality criterion
is concerned with join definitions occurring in parallel, \ie definitions
$D_1, \dots, D_k$ in processes of the form,
$$\jdef D_1 \jin P_1 \parl \dots \parl \jdef D_k \jin P_k\text.$$

\begin{definition}[Normality of \Jcore]
We call a process $P \in \Jcore$ {\em normal} if for all definitions $D, D'$
occurring in parallel in $P$, it holds that $\dv(D) \cap \dv(D') = \emptyset$.
\end{definition}

We define the semantics of core join processes by their labeled transition
systems respecting the reduction semantics given by Fournet \cite{Fournet1998}.
In \fig{join:semantics}, we extended Fournet's semantics by an extra
type of labeled arrows which represent the $\tau$-labeled steps in
Fournet's semantics. $\xlongrightarrow{D}$ describes potential steps over $D$,
while $\xlongmapsto{D}$ describes actual reaction steps. 
We extended the original semantics to
make the \LTS of join comparable to the labeled net semantics we
propose in \sect{sect:semantics}.

\begin{definition}[\LTS of \Jcore]
Let $P \in \Jcore$. The {\em labeled transition system of $P$} is
$$\LTS(P) := ( \Jcore, \longmapsto, P )$$ 
where 
$\longmapsto \subseteq \Jcore \times \mathcal D \times \Jcore$
is the smallest relation respecting the structural operational
semantics in \fig{join:semantics}.
\end{definition}

In general, this labeled transition system is infinite and has unreachable
parts. The \textsc{Join} rule reveals potential reactions. The actual
reaction rule, \ie \textsc{React}, introduces the new arrow type.
Only if $P$ has a potential $D$ step to $P'$, then the reaction actually
takes place. For the remaining rules we have one for the potential
arrows and one for the reaction arrow. The \textsc{Par} rules work as expected.
A join definition can be skipped if a reaction has already taken place,
\ie \textsc{Jump}$_2$, or the potential step $D$ does not interfere with
other free variables, \ie as in \textsc{Jump}$_2$.
The \textsc{Struct} rules refer to the structural congruences as defined
in \fig{join:struct}. For a better understanding of the labeled
transition semantics we give two examples.

\begin{example}
Consider the process $P = \jdef x\angled{u} \parl y\angled v \join 
u\angled v \jin x\angled k \parl x\angled j \parl y\angled 2$.
For simplicity, we use the definition variable $D = x\angled u \parl y\angled v \join u\angled v$.
Intuitively, $P$ has two possible executions. First, $x\angled k$ and $y\angled 2$
react with $D$ or second,
$x\angled j$ and $y\angled 2$ react under $D$. In both cases,
one message $x\angled \_$ remains in the process. As $x\angled j \parl y\angled 2$
potentially react with $D$, the rule \textsc{Join} tells that
$x\angled j \parl y\angled 2 \xlongrightarrow D j\angled 2$. Now, \textsc{React}
can be directly applied, \ie $\jdef D \jin x\angled k \parl x\angled j \parl y\angled 2
\xlongmapsto D \jdef D \jin x\angled k \parl j\angled 2$. From there on,
there is no other step possible. The second execution can be obtained 
by the use of \textsc{Struct1}. We needed $x\angled k$ and $y\angled 2$
in parallel. Due to commutativity and associativity of the parallel
operator, this is possible. Therefore, by \textsc{Struct1} we obtain
$x\angled k \parl x\angled j \parl y\angled 2 \xlongrightarrow D k\angled 2 \parl x\angled j$.
Again, we may apply \textsc{React} to get the actual reaction, \ie
$\jdef x\angled u \parl y \angled v \join u\angled v \jin x\angled k \parl x\angled j \parl y\angled 2
\xlongmapsto D \jdef x\angled u \parl y\angled v \join u\angled v \jin k\angled 2 \parl x\angled j$.
These are the only $\longmapsto$-steps. So, the \LTS of $P$ is a choice
between the message $j\angled 2$ and $k\angled 2$.
\end{example}

In the last example we already saw how \textsc{Join}, \textsc{React}
and \textsc{Struct} are applied. The application of the \textsc{Par}
rules is as expected. The next example considers a process where both 
\textsc{Jump} rules are applied.

\begin{example}
Consider
$$P = \jdef x\angled u \parl y\angled v \join u\angled v \jin
\jdef a\angled v \join v\angled{} \jin x\angled a \parl y\angled 2\text.$$
Again, we abbreviate the definitions occurring in $P$, \ie $D_1 = x\angled u \parl y\angled v \join u\angled v$
and $D_2 = a\angled v \join v\angled{}$. In a first step, we need to 
identify the potential steps of $P$. Considering $P$, there is only
one potential step that matters, namely $x\angled a \parl y\angled 2
\xlongrightarrow {D_1} a\angled 2$. With that knowledge we can apply
\textsc{Jump1}, because $\dv(D_1) \cap \fv(D_2) = \emptyset$. This yields
the following arrow, $\jdef D_2 \jin x\angled a \parl y\angled 2 \xlongrightarrow {D_1} \jdef D_2 \jin a\angled 2$.
The \textsc{React} rule does the rest, \ie $\jdef D_1 \jin \jdef D_2 \jin x\angled a \parl y\angled 2
\xlongmapsto {D_1} \jdef D_1 \jin \jdef D_2 \jin a\angled 2$. We are
almost done. The \textsc{React} rule exhibits the next arrow,
$\jdef D_2 \jin a\angled 2 \xlongmapsto {D_2} 2\angled{}$. To transfer this result
to the whole process, we apply \textsc{Jump2}, \ie $\jdef D_1 \jin \jdef D_2 \jin a\angled 2
\xlongmapsto {D_2} \jdef D_1 \jin \jdef D_2 \jin 2\angled{}$.
\end{example}

Note that this example is similar to the one at the beginning of \sect{sect:semantics}.
For discussions on the distributability of the join-calculus in \sect{sect:distributability}, 
we need to mention the notion of {\em locality}. In the join-calculus,
receptors must reside on one location, \ie they cannot be extruded to
more than one location. Therefore,
a join definition $J \join P$ can be seen as such a location and hence, a
location function is implicitly given in core join. We assume each join
definition appearing in a join process, either directly or by reduction,
to constitute a location. This is an approximation, because system modelers might summarize several join definitions
to one location. To express this freedom, a distributed version of the
join-calculus has been developed. The {\em distributed
join-calculus} \cite{Fournet1996a} employs explicit location functions and comes with a fully abstract encoding into the join-calculus.
However, we concentrate on the core join-calculus. For later discussions,
we rely on the above mentioned assumptions on locality.

%% file: semantics.tex
\section{Petri Net Semantics for Join}\label{sect:semantics}
The semantics operates in two steps. First, the join term is
decomposed into an initial set of places. Each place is equipped with
a message term of core join, \eg $x\angled{v}$, and
the scopes of $x$ and $v$, because both names may have their individual
scopes. \ex{example:scopes} shows the need for both scopes.
\begin{example}\label{example:scopes}
$$P = \jdef \underbrace{x\angled{v} \parl y\angled{w} \join v\angled{w}}_{D_1} \jin
\jdef \underbrace{a\angled v \join \nil}_{D_2} \jin x\angled a \parl y\angled 2\text.$$
In $P$, we have six names: $a, 2, x, y, v, w$. While $x$ and $y$ are defined
by $D_1$, $a$ is defined by $D_2$ and $2$ is free. The names $v$ and $w$
are received variables and do not occur in a message. Here $x\angled a$
has the same scope as $x$, but $a$ is scoped by $D_2$. So after a $D_1$
step, there is a message $a\angled 2$, which may react in $D_2$.
Therefore, each place is equipped with both, the scope of the sender
and the scope of the sent name.
\end{example}

The decomposition yields only places for message terms. Parallel
compositions and join definitions are represented in the net structure.

The second step of our semantics consists of applications of a transition
rule which makes use of the information stored in places. Given two 
places representing $x\angled a, y\angled 2$ in the example above, our 
transition rule ensures that there exists a transition, labeled by $D_1$, 
consuming from both places and producing to places that correspond to the 
right side of the reaction rule, \ie the decomposition of $v\angled{w}$, 
where $v$ is mapped to $a$ and $w$ to $2$. The just described decomposition
yields a place $a\angled 2$ which can react in $D_2$ producing no new messages. 
The Petri net representation of \ex{example:scopes} is depicted in 
\fig{example:app_semantics}.

\begin{figure}[h]
\begin{center}
\begin{tikzpicture}
\node[place,tokens=1] (xa) [label=below left:$x\angled a$] {};
\node[place,tokens=1] (y2) [below=of xa,label=below left:$y\angled 2$] {};
\node (h1) [right=of xa] {};
\node[transition,node distance=20pt,below=of h1] (D1) {$D_1$};
\node[place] (a2) [right=of D1,label=below:$a\angled 2$] {};
\node[transition] (D2) [right=of a2] {$D_2$};

\draw[edge] (xa) to (D1);
\draw[edge] (y2) to (D1);
\draw[edge] (D1) to (a2);
\draw[edge] (a2) to (D2);
\end{tikzpicture}
\end{center}
\caption{Petri net semantics of $P$ in \ex{example:scopes}}\label{example:app_semantics}
\end{figure}

Note that, although we exploit the parallel structure of a process, join
definition applications are only unfolded. Therefore, our semantics yields
in general infinite net representations.

\subsection{Operational Semantics}\label{sect:pnsemantics}
Our Petri net definitions in \sect{sect:petrinets} left two main points
open, which need to be defined in advance. First, the universe of places $\mathcal P$ and
second, the set of transition labels $\Sigma$.
As already mentioned in the last section, places are
triples. The first component is a join message, \eg $x\angled v$.
The second and third components are stacks over the set of join definitions
$\mathcal D$. The first stack represents the scope of the sender name,
the second stack that of the sent name.
$$\mathcal P := \{ x\angled v \,|\, x, v \in \mathcal N \} \times \mathcal S_{\mathcal D} \times \mathcal S_{\mathcal D}$$
denotes the universe of places.
Labels for transitions are join definitions, \ie $\Sigma := \mathcal D$.
In \fig{example:app_semantics} we have labeled each place with the
message it represents.

The decomposition function returning sets of places for core join terms
needs to be equipped with an auxiliary function to manage the name scoping.
In the following, such functions are referred to as $f$ or $f_\bot$.
$f$ maps names in $\mathcal N$ to names in $\mathcal N$ and stacks over
$\mathcal D$, \ie $f : \mathcal N \to ( \mathcal N \times \mathcal S_{\mathcal D} )$.
For $n \in \mathcal N$, $f^1(n)$ represents a certain renaming of $n$ 
(\cf \defn{def:transitionrule}). $f^2(n)$ stores the scope of $f^1(n)$.
Initially, we use $f_\bot$ with $f_\bot(n) := ( n, \bot )$ ($n \in \mathcal N$).

During the application of the decomposition, it is necessary to alter the 
scopes for names. For this purpose, we use a special function $g_n$ 
operating on any  
$f : \mathcal N \to ( \mathcal N \times \mathcal S_{\mathcal D} )$. 
This function shall reduce the stack of $n$ by one element.
$g_n(f) : \mathcal N \to ( \mathcal N \times \mathcal S_{\mathcal D} )$
works like $f$ if the parameter is not $n$. Otherwise, it returns what 
$f$ returns, but the stack component is reduced by one element, \ie
$$(g_n(f))(x) := \left\{\begin{array}{lcl}
( \id \times \pop ) \circ f(x) && x = n\text, \\
f(x) && \text{otherwise.}
\end{array}\right.$$

The decomposition function $dec$ is defined inductively over the
structure of core join processes.

\begin{definition}
The function $dec : ( \Jcore \times ( \mathcal{N} \to ( \mathcal N \times \mathcal{S}_\mathcal{D} ) ) ) \to 2^\mathcal{P}$ is called {\em decomposition
function}. For all $x, v \in \mathcal N$, $P, Q \in \Jcore$, $D \in \mathcal D$, and $f : ( \mathcal{N} \to ( \mathcal N \times \mathcal{S}_\mathcal{D} ) )$
the decomposition is defined by
$$\begin{array}{rcl}
(\nil, f) & \mapsto & \emptyset\text, \\\\
( x\angled{v}, f ) & \mapsto & \left\{\begin{array}{lcl}
dec( x\angled{v}, g_x(f) ) && f^1(x) \not\in \dv(f^2(x)\top) \\
dec( x\angled{v}, g_v(f) ) && f^1(v) \not\in \dv(f^2(v)\top) \\
\{ ( f^1[x\angled v], f^2(x), f^2(v) ) \} && \text{otherwise,}
\end{array}\right.\\\\
( P \parl Q, f ) & \mapsto & dec( P, f ) \uplus dec( Q, f )\text, \\\\
( \jdef{D} \jin {P}, f ) & \mapsto & dec( P, (\id \times \push D) \circ f )\text.
\end{array}$$
\end{definition}

Note that the decomposition always yields finite sets of places.
The \nil process yields the empty set of places. The result of the
decomposition also corresponds to markings. Here, the empty marking represents
exactly what we expect from the behavior of \nil, \ie no behavior.
The decomposition of the parallel operator is represented by the disjoint 
union of both components. So, even two equal messages running in
parallel are decomposed into two places. Therefore, we use the
equality symbol $=$ as {\em equality up to isomorphism}, when we refer to
decompositions or markings of the resulting nets, respectively.
In the decomposition of join definitions, we need to adjust the
renaming function $f$, which also handles the scoping of names. A
join definition is decomposed as $P$, but the renaming function is
extended by $(\id \times \push D)$, meaning, that each name now has
a new scope, in particular $D$ and all other definitions which were
already stored in $f$.

The decomposition of messages $x\angled v$ does the main work, because
it handles the scopes of $x$ and $v$. By several applications of $g_x$
and $g_v$, it assigns the correct scopes to the resulting place.
Note that we assume $n \in \dv(\bot)$ for all $n \in \mathcal N$.

The recursive application of $dec$ eventually terminates, because in
each step, the terms in the decomposition get smaller. Either a
parallel operator or a join definition is removed. Decompositions of
messages also terminate, as the stacks for sender and sent name are
reduced by one element as long as they are not empty or the queried
name occurs in the set of defined variables. One of the
two possibilities holds eventually.

Given a core join process $P$. The decomposition of $P$ yields the set
of initially marked places. The behavior of $P$ is not mapped to the
semantics yet. Instead of giving an algorithm to construct 
a net, we give a rule that must be satisfied by a Petri net to be 
the semantics of $P$. To reflect the labeled transition semantics
of the core join-calculus, we need to ensure that definitions can be
applied, \ie transitions may fire, if their preconditions are satisfied. 
Definitions have the form
${x\angled{u}\parl y\angled{v}}\join{Q}$, where a process must be able
to send messages over $x$ and $y$ to perform the definition, \ie create
a new process $Q$ instantiated with the received variables. As our places
carry the necessary scoping, we use that information in \defn{def:transitionrule}.
A transition consuming from the preconditions of a join definition it
represents is forced to produce to places to which another transition
does not produce. By this, we reach that places never branch backwards,
an important condition discussed later in \sect{sect:structure}. Furthermore,
a transition must not produce to the initially marked places. By this,
we obtain an acyclic structure, \ie bounded places. Indeed, the transition
rule and the nature of our decomposition function ensure our Petri net
semantics to yield {\em 1-safe Petri nets}.

\begin{definition}\label{def:transitionrule}
Let $N = ( P, T, m_0 )$ be a labeled Petri net over $(\mathcal P, \mathcal D)$.
$N$ satisfies the {\em transition rule} iff for every two places $p, q \in P$ with 
\begin{itemize}
\item $p = ( x\angled a, s, s_a )$, $q = ( y\angled b, s, s_b )$ and 
\item $s\top = {x\angled{u}\parl y\angled{v}} \join {R}$,
\end{itemize}
it holds that there exists a transition $t \in T$ with
\begin{itemize}
\item $t = ( \{ p, q \}, x\angled{u} \parl y\angled{v} \join R, P' )$,
\item $P' \cap m_0 = \emptyset$ and $\pre P' = \{ t \}$
\end{itemize}
where $P' = dec(R, f_t)$ and
$f_t : \mathcal N \to ( \mathcal N \times \mathcal S_{\mathcal D} )$
with for $n \in \mathcal N$
$$f_t(n) = \left\{\begin{array}{lcl}
( a, s_a ) && n = u\text, \\
( b, s_b ) && n = v\text, \\
( n, s ) && \text{otherwise.}
\end{array}\right.$$
\end{definition}

In the transition rule, renamings encoded in $f_t$ become important.
As it is possible to have equal names with different scopes, a reaction,
\ie a transition in our nets, needs to respect the scopes although the
names are equal. Therefore, we postponed the renaming in the decomposition
function to the end of the procedure. Consider \ex{example:renaming}
as an illustration.

\begin{example}\label{example:renaming}
$$Q = \jdef a\angled{k} \parl b\angled{k'} \join k\angled{} \parl k'\angled{} \jin b\angled{c} \parl \jdef c\angled{} \join \nil \jin a\angled{c}\text.$$
$Q$ contains two names $c$ with different scopes. The $c$ sent over $b$
is free in $Q$. The $c$ sent over $a$ is defined. Our construction
respects both $c$s via $f_t$. Instead of renaming the resulting process,
here $k\angled{} \parl k'\angled{}$, to $c\angled{} \parl c\angled{}$
first, we decompose the right side of a join definition and
apply the necessary renaming afterward. Therefore, our semantics is
able to distinguish both variables $c$.
\end{example}

Given a core join process $J$. To construct the Petri net semantics for
$J$, we begin with the set of initially marked places. This set corresponds
with the initial decomposition, \ie $dec( J, f_\bot )$. If there are no
applicable definitions in $J$, the net construction is finished. Otherwise, there must be
at least two places violating the just defined transition rule. In order
to satisfy the transition rule, we add a transition and a set of places as described
in \defn{def:transitionrule}. We repeat this procedure until the
net satisfies the transition rule. The resulting Petri net represents
the semantics of $J$.

\begin{definition}\label{Nj}
Let $J \in \Jcore$ be some core-join process. The Petri net
$N(J) = ( P, T, m_0 )$ represents the semantics of $J$ if it
is the smallest Petri net satisfying
\begin{enumerate}
\item $m_0 = dec(J, f_\bot) \subseteq P$ and
\item the transition rule.
\end{enumerate}
\end{definition}

In this section, we have already seen an example (\ex{example:scopes})
and its Petri net semantics in \fig{example:app_semantics}.
Note that the procedure described above yields exactly those nets
satisfying \defn{Nj}. The criterion asking for the {\em smallest} net
ensures that dead transitions and isolated places are left out.

\subsection{Structural Properties}\label{sect:structure}
In this section, we investigate the net class of our Petri net semantics,
\ie {\em 1-safe} Petri nets. This net class restricts all places to
contain at most one token for any reachable marking, especially the
initial marking. As our decomposition function relies on disjoint unions,
initial markings in our nets {\em are 1-safe}.

In order to show the net class, we prove the following properties, 
also valid for {\em occurrence nets} \cite{Nielsen1979}.

\begin{proposition}\label{prop:occurrence}
Let $J \in \Jcore$ be a process. $N(J) = ( P, T, m_0 )$
satisfies the three criteria below.
\begin{enumerate}
\item For all $p \in m_0$ it holds that $\pre p = \emptyset$.
\item For all $p \in P$ it holds that $|\pre p| \leq 1$.
\item $F^+$ (transitive closure of $F$) is irreflexive.
\end{enumerate}
\end{proposition}

The first property states that there are no transitions in the net
producing tokens to initially marked places in $m_0$. The second states 
that there is always one and only one reason, \ie a transition,
that produces a token to a place. The last one is concerned with cycles
in the net structure.

\begin{proof}
Let $J \in \Jcore$ be a process and $N(J) = ( P, T, m_0 )$ its
\PN semantics.
\begin{enumerate}
\item\label{prf:one} We need to show that for all initially marked places,
\ie $p \in m_0$, it holds that their presets are empty. As $N(J)$ needs
to fulfill the transition rule (\defn{def:transitionrule}), there is no 
transition $t \in T$ with $\pre t \neq \emptyset$ and $\post t \cap m_0 \neq \emptyset$.
If there are transitions $t$ with $\pre t = \emptyset$ producing to
$m_0$, then $N(J)$ is not the smallest net after \defn{Nj}. Therefore,
there is no transition producing the $m_0$ and in consequence, the claim holds.
\item\label{prf:two} We need to show that for all places $p \in P$,
there is at most one transition $t \in \pre p$. By \defn{def:transitionrule}, $N(J)$
needs to satisfy the transition rule. From \ref{prf:one} we know
that the claim holds for initially marked places. For any other place $p$,
we need to show that there are no two transition $t, t' \in T$ with
$p \in \post t \cap \post {t'}$. From the transition rule we follow
that $P_1 = \post t$ and $P_2 = \post {t'}$. The transition rule also
ensures that $\pre {P_1} = \{ t \}$ and $\pre {P_2} = \{ t' \}$. If
$p$ was in $P_1$ and in $P_2$, then $P_1 = P_2$ and in consequence 
$t = t'$. Therefore, $|\pre p| \leq 1$.
\item We need to show that there are no cycles in our net representations.
By the net construction, we prove that our nets do not introduce cycles.
Starting with the set of initial places, the transition rule can only
introduce transitions producing to places which are not initially marked.
Otherwise, this would contradict \ref{prf:one}. Let $p$ be an arbitrary
place in the net. From some place in $m_0$ to $p$ are no cycles in the net.
Let $Q$ be the set of all places between  $m_0$ and $p$.
A transition $t$ consuming from $p$ produces to a set of places $P'$.
We need to show that $P'$ is disjoint from $Q$. Assuming, $P' \cap Q \neq \emptyset$. So, there
is a place $q \in Q$ which is also in $P'$. $q$ cannot be in the set
of initially marked places. Therefore, there exists a transition $t_q$
producing to $q$. Now, $\pre q = \{ t_q, t \}$ which contradicts \ref{prf:two},
unless $t \neq t_q$. Therefore, $F^+$ is irreflexive.
\end{enumerate}
\end{proof}

\prop{prop:occurrence} enables us to show that our net semantics 
produces 1-safe Petri nets. We use the fact that 
$\max\{ m_0(p) ~|~ p \in P \} \leq 1$, for all $J \in \Jcore$ with
$N(J) = ( P, T, m_0 )$. Furthermore, we have already
proven that there are no cycles in our net semantics and for each
place, there is at most one transition producing to it. Therefore, we
can formulate the following corollary.

\begin{corollary}\label{coroll:1safe}
Let $J \in \Jcore$ be a process. Then $N(J)$ is 1-safe.
\end{corollary}

The proof follows directly from \prop{prop:occurrence}. For further
discussions we introduce the notions of causality, conflict and
independence on the basis of Petri nets.

\begin{definition}
Let $N = (P, T, m_0)$ be a Petri net and $t_1, t_2 \in T$. $t_1$ and
$t_2$ are said to be in {\em causal order}, $t_1$ before $t_2$, iff
there is a reachable marking $m_1$ with $m_1[t_1\rangle m_2$ and a reachable
marking $m_3$ from $m_2$ with $m_3[t_2\rangle$ but no such markings which
enable $t_2$ first. $t_1$ and $t_2$ are {\em in direct conflict} iff 
$\pre t_1 \cap \pre t_2 \neq \emptyset$. Two nodes $n_1, n_2 \in P \cup T$
are {\em in conflict} iff there exist two transitions $t, t' \in T$
which are in conflict and there exist paths from $t$ to $n_1$ and from
$t'$ to $n_2$. If $n_1 = n_2$, then $n_1$ is in {\em self-conflict}.
$t_1$ and $t_2$ are {\em independent} (or {\em concurrent}) iff
they are neither in a causal order nor in conflict.
\end{definition}

Intuitively, the notion of independence describes actions, \ie transitions,
which can always occur in parallel.
There is a
remaining property of occurrence nets which is not satisfied by our nets, namely
irreflexivity of the conflict relation. This property states that there
are no self-conflicting nodes in the net.

The join-calculus semantics relies on the structural congruences of
\fig{join:struct}. Therefore, our net semantics needs to reflect them
in a proper way. Indeed, there is a provable
correspondence between the structural congruences of the core join-calculus
and the Petri net representations. We prove that if two join terms
are structurally congruent, then their net representations are isomorphic.

\begin{lemma}\label{lemma:iso}
Let $P, Q \in \Jcore$ be processes with $P \equiv Q$. Then $N(P)$ and
$N(Q)$ are isomorphic.
\end{lemma}

The proof can be found in the technical report to the paper \cite{MennickeTR2012}. \lemm{lemma:iso}
also has a side effect to the following behavioral correspondence. We
will show a bisimulation between core join terms and their net representations.
One of the proof steps is concerned with structurally congruent join terms.
As isomorphisms imply bisimulation \cite{Glabbeek2001}, we can assume it as already proven
by \lemm{lemma:iso}.

\subsection{Behavioral Properties}\label{sect:correctness}
In this section, we will prove that the semantics we presented is correct
with respect to bisimulation. We already saw \LTS interleaving semantics
for both, \PNs and the join-calculus. The states of an \LTS for a \PN
is described by markings. States of core join \LTS are core join terms.
We need to find a bisimulation $\mathcal R \subseteq \Jcore \times 2^{\mathcal P}$.
Note that any subset of $\mathcal P$ describes a valid marking of a
Petri net of a core join term.

Our bisimulation result relies on the observation, that our decompositions
yield valid markings of a net describing the semantics of a core join term. Each state of a process $P$ is represented
by its initial decomposition $dec(P, f_\bot)$. When $P$ evolves to $P'$,
then our Petri net semantics reflects this behavior by a step from
$dec(P, f_\bot)$ to $dec(P', f_\bot)$, because all join definitions of
$P$ are preserved by $P'$ and so, they remain on some stack in the
decomposition of $P'$. Conversely, if our net evolves
from $dec(P, f_\bot)$ to $m$, then this $m$ must be equivalent to some
$dec(P', f_\bot)$, \ie there is a step from $P$ to $P'$. We need to 
prove that this is actually true for all $P \in \Jcore$.

Using the just described observation, we formulate a base bisimulation
as follows,
$$\mathcal R := \left\{ \left( P, dec( P, f_\bot ) \right) \,|\, P \in \Jcore \right\}\text.$$
When considering a process $P$, then we restrict $\mathcal R$ to the
reachable parts of $P$, denoted by $R_P := \mathcal R \upharpoonright_{P \longmapsto^*}$.

\begin{theorem}\label{theorem:interleaving}
Let $P \in \Jcore$. Then $\LTS(P)$ and $\LTS(N(P))$ are bisimilar.
\end{theorem}

The proof can be found in the technical report to this paper \cite{MennickeTR2012}.

%% file: discussion.tex
\section{Distributability Issues in the Join Calculus}\label{sect:distributability}
One of the advantages of Petri net semantics for process calculi is
the inherent notion of independence. A set of independent actions, \ie
the labels of independent transitions, is called a {\em step}.
A step is enabled if all its transitions are enabled. An enabled step may
fire. The resulting marking is the same marking as if all transitions
in a step fired in a sequence. Therefore, if we consider an \LTS
construction in terms of Petri net steps, we do not get more states, but more transitions,
because independent actions are summarized in multisets.

The induced steps on the semantics of the core join-calculus correspond
to independent join definition applications.
Chains of join definitions are are translated into sequences of transitions.
Our net semantics also recognizes definition chains which are actually
independent, due to the fact that our Petri net semantics respects
the structural congruence (\cf \lemm{lemma:iso}).

\begin{figure}
\begin{center}
\begin{tikzpicture}
\node[place,tokens=1,label=left:$p$] (p) at (-1,0) {};
\node[place,tokens=1,label=left:$q$] (q) at (1,0) {};
\node[transition,label=below:$t_1$] (a) at (-2,-1.5) {};
\node[transition,label=below:$t_2$] (b) at (0,-1.5) {};
\node[transition,label=below:$t_3$] (c) at (2,-1.5) {};

\draw[edge] (p) to (a);
\draw[edge] (p) to (b);
\draw[edge] (q) to (b);
\draw[edge] (q) to (c);
\end{tikzpicture}
\end{center}
\caption{The confusion pattern $\mathsf{M}$}\label{fig:M}
\end{figure}
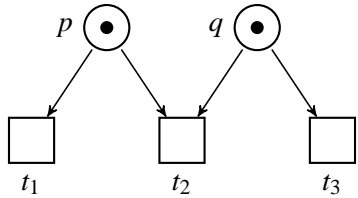

Steps enable our semantics to argue about the distributability of the
join-calculus, or more precisely, about the distributability of our
net representations of the join-calculus. We are interested in a particular
{\em confusion} pattern which is depicted in \fig{fig:M}. This structure
is called \Mstruct. The \Mstruct was introduced by van Glabbeek \etal
as a structure which has a major influence to the {\em distributability} of a system
\cite{Glabbeek2008,Glabbeek2009}. A net is distributable if there exists
a behaviorally equivalent net which is {\em distributed}. Van Glabbeek \etal
call a system distributed if
\begin{itemize}
\item it consists of components on different locations,
\item the components work concurrently,
\item the components interact explicitly, and
\item communication between components is asynchronous.
\end{itemize}
They formalized those criteria in a Petri net class called {\em LSGA nets} --
locally sequential, globally asynchronous nets. The crucial point of
LSGA nets is that parallel transitions are not allowed to be on one location
while transitions sharing input places must share one.
The \Mstruct is not distributed as all transitions need to reside on one location,
but $t_1$ and $t_3$ may fire in a step, \ie in parallel. Van Glabbeek \etal
proved that if a net contains a fully reachable \Mstruct, \ie there is
a reachable marking containing at least the places in \fig{fig:M}, then the Petri net is
not distributable up to branching-time equivalences \cite{Glabbeek2008}.
Schicke-Uffmann \etal prove that the \Mstruct is not distributable
in terms of causality respecting equivalences \cite{Schicke2011}.
Their arguments depend on the chosen notion of distributed systems and
distributability. However, we consider these notions as reasonable,
because the described points above are important phenomena occurring in distributed
system design and implementation.

Therefore, if we identify such a structure in our net semantics, there
is a potential restriction on the distributability of the join-calculus, \ie
join-calculus processes.

\begin{example}\label{example:M}
Consider the following process,
$$P = \jdef \underbrace{x\angled u \parl y\angled v \join u\angled v}_{D} \jin x\angled a \parl y \angled 1 \parl x\angled b \parl y\angled 2\text.$$

The \PN semantics $N(P)$ of process $P$ is depicted in \fig{fig:jM}.
$N(P)$ contains four \Mstruct{s} as depicted in \fig{fig:M}.
Initially, the process makes a choice between four different
join definition applications. After one application, there is only
one possibility for the resulting process to apply the join definition
again. Our \PN semantics reflects this behavior.

\begin{figure}
\begin{center}
\begin{minipage}{13cm}
\begin{center}
\begin{tikzpicture}
\node[place,tokens=1] (xa) at (0,1.5) [label=above:$x\angled{a}$] {};
\node[place,tokens=1] (y1) at (-2,0) [label=above right:$y\angled{1}$] {};
\node[place,tokens=1] (xb) at (0,0) [label=above:$x\angled{b}$] {};
\node[place,tokens=1] (y2) at (2,0) [label=above left:$y\angled{2}$] {};

\node[place] (a1) at (-3,-3) [label=left:$a\angled 1$] {};
\node[place] (b1) at (-1,-3) [label=left:$b\angled 1$] {};
\node[place] (b2) at (1,-3) [label=left:$b\angled 2$] {};
\node[place] (a2) at (3,-3) [label=left:$a\angled 2$] {};

\node[transition] (t1) at (-3,-1.5) {$D$};
\node[transition] (t2) at (-1,-1.5) {$D$};
\node[transition] (t3) at (1,-1.5) {$D$};
\node[transition] (t4) at (3,-1.5) {$D$};
  
\draw[edge,bend right] (xa) to (t1);
\draw[edge] (y1) to (t1);
\draw[edge] (t1) to (a1);

\draw[edge] (y1) to (t2);
\draw[edge] (xb) to (t2);
\draw[edge] (t2) to (b1);

\draw[edge] (y2) to (t3);
\draw[edge] (xb) to (t3);
\draw[edge] (t3) to (b2);

\draw[edge,bend left] (xa) to (t4);
\draw[edge] (y2) to (t4);
\draw[edge] (t4) to (a2);
\end{tikzpicture}
\end{center}
\end{minipage}
\end{center}
\caption{Petri net semantics of $P$ in \ex{example:M}.}\label{fig:jM}
\end{figure}
\end{example}

The net semantics of the process in \ex{example:M} yields an \Mstruct 
just like the one in \fig{fig:M}. It is fully reachable, as
the initial marking enables all four \Mstruct{s}. We observe that all transitions
are labeled by the same definition.
Considering the notion of locality for the join-calculus (\cf \sect{sect:join}),
this structure remains on one location, although it contains
independent transitions. This fact makes a distributability result
of the join-calculus incomparable to the results in \cite{Glabbeek2008},
because van Glabbeek \etal forbid such structures on one location.
On the other hand, the implicit location function given by join definitions
gives reason to extend the notion of distributability.

In the following, we refer to an \Mstruct where all transitions are
labeled by the same definition as {\em local}.
If all \Mstruct{s} in the join-calculus were local, then the join-calculus
would be a distributable process calculus, because our net semantics respects
the behavior of the join-calculus and van Glabbeek \etal prove that
a Petri net with no fully reachable \Mstruct is distributable \cite{Glabbeek2012}.
The following proposition gives proof for this hypothesis.

\begin{proposition}
Let $J \in \Jcore$. If $N(J)$ contains a fully reachable \Mstruct,
then it is local.
\end{proposition}
\begin{proof}
We prove the claim by contradiction. Let $J \in \Jcore$ be a process and $N(J) = ( P, T, m_0 )$
be the Petri net semantics of $J$. Assuming \fig{fig:M} is a part of $N(J)$
and each transition has a different label, \ie $l(t_i) \neq l(t_j)$ for
$i \neq j$ and $i, j = 1, 2, 3$. From the transition rule, it follows that all preplaces
of a transition have the same stack in their second component. Especially,
the top element of these stacks is equal to the label of the transition.
Reconsider \fig{fig:M}. As $p \in \pre t_1$, we know that $p = ( \_, s, \_ )$
with $s\top = l(t_1)$. $p \in \pre t_2$, so $p = ( \_, s', \_ )$ with
$s'\top = l(t_2)$. But, by construction, this is not possible if $l(t_1) \neq l(t_2)$.
Therefore, either $t_1, t_2$ do not exist or $l(t_1) = l(t_2)$.
The case of $t_2, t_3$ is analogous, \ie $l(t_2) = l(t_3)$. By transitivity,
we have $l(t_1) = l(t_3)$.
\end{proof}

It is not possible to have an \Mstruct with different transition 
labels, \ie on different locations, in the join-calculus. The proof steps make use of a property
of the join-calculus which is reflected by our \PN semantics. This
property is concerned with the assignment of messages to join definitions,
\ie the number of transitions with different labels in the postset of
a place. For each join message, there is at most one applicable join
definition.

Van Glabbeek \etal \cite{Glabbeek2008,Glabbeek2009,Glabbeek2012} and
Schicke-Uffmann \etal \cite{Schicke2011} consider
unlabeled nets with no explicit location function to derive their distributability
results. If we consider the join-calculus as a distributable process
calculus, then it is a natural step to evaluate their results given the
assumptions of the join-calculus. Best and Darondeau \cite{Best2011}
already consider a given allocation function in their survey paper to
argue on the distributability of Petri nets.

%% file: conclusion.tex
\section{Conclusion}\label{sect:conclusion}
In this paper we presented an operational Petri net semantics for the
join-calculus. We proved that our semantics corresponds to structural
congruences and the labeled reduction semantics of the calculus.
Furthermore, we investigated issues of distributability in the join-calculus.

In future work, we want to understand how an explicit location function,
as implied by the join-calculus, influences the results of 
\cite{Glabbeek2008,Schicke2011}. Moreover, we would like to investigate 
optimizations of the semantics to possibly reach finite net 
representations of join terms.
The mentioned applications in unfolding based techniques is not discussed
in this paper. As our suggested semantics has an unfolding nature, it
is worthwhile to apply such techniques to the join-calculus by first
using our semantics to compute the necessary prefixes of a join term.

\paragraph{Acknowledgments.} The author gratefully thanks the DFG
(German Research Foundation) for financial support. Moreover, he wishes 
to thank Malte Lochau and the anonymous reviewers for their useful 
comments on the paper. Further acknowledgments go to Ursula Goltz, 
Uwe Nestmann, Kirstin Peters, and Jens-Wolfhard Schicke-Uffmann for 
valuable discussions.